\documentclass[aps,pra,groupedaddress,showpacs,twocolumn,superscriptaddress]{revtex4-1}

\usepackage{hyperref}
\usepackage{amsmath}
\usepackage{amssymb}
\usepackage{amsthm}
\usepackage{graphics}
\usepackage{graphicx}

\newcommand{\ket}[1]{|#1\rangle}               
\newcommand{\bra}[1]{\langle #1|}              
\newcommand{\dyad}[2]{\ket{#1}\bra{#2}}        
\newcommand{\ip}[2]{\langle #1|#2\rangle}      

\newcommand{\ii}{\mathrm{i}}
\newcommand{\ee}{\mathrm{e}}
\newcommand{\expo}[1]{\mathrm{e}^{#1}} 

\newtheorem{theorem}{Theorem}
\newtheorem{lemma}{Lemma}

\newtheorem*{QGS}{Quadratic Gauss Sums}
\newtheorem*{LS}{Landsberg-Schaar Identity}
\newtheorem*{GQGS}{Generalized Quadratic Gauss Sums}
\newtheorem*{RFGQGS}{Reciprocity Formula for Generalized Quadratic Gauss Sums}

\long\def\ca#1\cb{} 	

\newcommand{\HC}{\mathcal{H}}

\newcommand{\CP}{\mathrm{CP}}

\begin{document}

\title{Construction of Equientangled Bases in Arbitrary Dimensions via \mbox{Quadratic Gauss Sums} and \mbox{Graph States}}

\author{Vlad Gheorghiu}
\email{vgheorgh@andrew.cmu.edu}
\affiliation{Department of Physics, Carnegie Mellon University, Pittsburgh,
Pennsylvania 15213, USA}

\author{Shiang Yong Looi}
\email{slooi@andrew.cmu.edu}
\affiliation{Department of Physics, Carnegie Mellon University, Pittsburgh,
Pennsylvania 15213, USA}

\date{Version of June 26, 2010}

\begin{abstract}
Recently Karimipour and Memarzadeh [Phys. Rev. A \textbf{73}, 012329 (2006)] studied the problem of finding a  family of orthonormal bases in a bipartite space each of dimension $D$ with the following properties: (i) The family continuously interpolates between a product basis and a maximally entangled basis as some parameter $t$ is varied, and (ii) for a fixed $t$, all basis states have the same amount of entanglement. The authors derived a necessary condition and provided explicit solutions for $D \leqslant 5$ but the existence of a solution for arbitrary dimensions remained an open problem. We prove that such families exist in arbitrary dimensions by providing two simple solutions, one employing the properties of quadratic Gauss sums and the other using graph states. The latter can be generalized to more than two parties -- we can construct equally entangled families that vary continuously between a product basis and a graph state basis.
\end{abstract}

\maketitle

\section{Introduction\label{sct1}}
We present two different solutions to the problem posed by Karimipour and Memarzadeh in \cite{PhysRevA.73.012329} of constructing an orthonormal basis of two qudits with the following properties: (i) The basis continuously changes from a product basis (every basis state is a product state) to a maximally entangled basis (every basis state is maximally entangled), by varying some parameter $t$, and (ii) for a fixed $t$, all basis states are equally entangled. As mentioned in \cite{PhysRevA.73.012329}, such a family of bases may find applications in various quantum information protocols including quantum cryptography, optimal Bell tests, investigation of the enhancement of channel capacity due to entanglement and the study of multipartite entanglement. For a more detailed motivation the interested reader may consult \cite{PhysRevA.73.012329}.

The paper is organized as follows : In Sec.~\ref{sct2} we summarize the main results of \cite{PhysRevA.73.012329} and then introduce the concept of Gauss sums and some useful related properties. Next we provide an explicit parametrization of a family of equientangled bases and we prove that it interpolates continuously between a product basis and a maximally entangled basis, for all dimensions. We illustrate the behaviour of our solution with explicit examples. In Sec.~\ref{sct3} we construct another such family using a completely different method based on graph states, describe a simple extension of it to multipartite systems, and then illustrate its behaviour with specific examples. Finally in Sec.~\ref{sct4} we compare the two solutions and make some concluding remarks.

\section{Construction based on Gauss sums\label{sct2}}
\subsection{Summary of previous work\label{sct2A}}
Let us start by summarizing the main results of \cite{PhysRevA.73.012329}. Consider a bipartite Hilbert space $\HC\otimes\HC$, where both Hilbert spaces have the same dimension $D$. The authors first defined an arbitrary normalized bipartite state
\begin{equation}
\label{eqn1}
\ket{\psi_{0,0}}=\sum_{k=0}^{D-1}a_k\ket{k}\ket{k}.
\end{equation}
Next for $m,n=0,1,\ldots, D-1$, they considered the collection of $D^2$ ``shifted" states
\begin{align}
\ket{\psi_{m,n}} &=X^{m}\otimes X^{m+n}\ket{\psi_{0,0}} \notag\\
&=\sum_{k=0}^{D-1}a_k\ket{k\oplus m}\ket{k\oplus m\oplus n}, \label{eqn2}
\end{align}
where 
\begin{equation}\label{eqn3}
X:=\sum_{k=0}^{D-1}\dyad{k\oplus 1}{k}
\end{equation}
is the generalized Pauli (or shift) operator and $\oplus$ denotes addition modulo $D$. 
They noted that all states have the same value of entropy of entanglement \cite{NielsenChuang:QuantumComputation} given by the von-Neumann entropy
\begin{equation}\label{eqn4}
E(\ket{\psi_{m,n}})=E(\ket{\psi_{0,0}})=-\sum_{k=0}^{D-1}|a_k|^2\log_D|a_k|^2,
\end{equation}
where the logarithm is taken in base $D$ for normalization reasons so that all maximally entangled states have entanglement equal to one regardless of $D$. 

Demanding the states in \eqref{eqn2} be orthonormal yields
\begin{equation}\label{eqn5}
\sum_{k=0}^{D-1} (a_k)^{*}a_{k\oplus m}=\delta_{m,0},\text{ }\forall m=0,\ldots,D-1,
\end{equation}
and the authors proved (see their Eqn. (36)) that \eqref{eqn5} is satisfied if and only if the coefficients $a_k$ have the form
\begin{equation}\label{eqn6}
a_k=\frac{1}{D}\sum_{j=0}^{D-1}\ee^{\ii \theta_j}\omega^{kj},
\end{equation}
where $\theta_j$ are arbitrary real parameters and $\omega=\expo{2\pi\ii/D}$ is
 the $D$-th root of unity. 
 
Therefore the authors found a family of $D^2$ orthonormal states, all having the same Schmidt coefficients and hence the same value of entanglement.  To ensure it interpolates from a product basis to a maximally entangled basis, it is sufficient to find a set of parameters $\{\theta^0_j\}_{j=0}^{D-1}$ for which the magnitude of $a_k$ is  $|a_k|=1/\sqrt{D}$ for all $k$. Then the problem is solved by defining
\begin{equation}\label{eqn7}
a_k(t):=\frac{1}{D}\sum_{j=0}^{D-1}\ee^{\ii t\theta^0_j}\omega^{kj},
\end{equation}
where $t\in[0,1]$ is a real parameter. When $t=0$ we have $a_k=\delta_{k,0}$ so the basis states are product states and when $t=1$, the basis is maximally entangled by assumption. We also observe there is a continuous variation in between these two extremes as a function of $t$. 

Karimipour and Memarzadeh considered the existence of such a set $\{\theta^0_j\}_{j=0}^{D-1}$ in arbitrary dimensions (see the last paragraph of Sec. V in \cite{PhysRevA.73.012329}). They found particular solutions for $D\leqslant 5$, but did not find a general solution for arbitrary $D$.

\subsection{Quadratic Gauss Sums\label{sct2B}}
We now define the basic mathematical tools we will make use of in the rest of this section. The most important concept is that of a \emph{quadratic Gauss sum}, defined below.
\begin{QGS}
Let $p,m$ be positive integers. The quadratic Gauss sum is defined as
\begin{equation}\label{eqn8}
\sum_{j=0}^{p-1}\expo{2\pi\ii j^2 m/p}.
\end{equation}
\end{QGS}
The quadratic Gauss sums satisfy a reciprocity relation known as 
\begin{LS}
Let $p,m$ be positive integers. Then
\begin{equation}\label{eqn9}
\frac{1}{\sqrt{p}}\sum_{j=0}^{p-1}
\expo{2\pi\ii j^2m/p}=
\frac{\ee^{\pi\ii/4}}{\sqrt{2m}}\sum_{j=0}^{2m-1}\expo{-\pi\ii j^2 p/2m}
\end{equation}
\end{LS}
The quadratic Gauss sums can be generalized as follows.
\begin{GQGS}
Let $p,m,n$ be positive integers. The generalized quadratic Gauss sum is defined as
\begin{equation}\label{eqn10}
\sum_{j=0}^{p-1}\expo{2\pi\ii( j^2 m + j n )/p}.
\end{equation}
\end{GQGS}
Finally the following reciprocity formula for generalized Gauss sums holds.
\begin{RFGQGS}
Let $p,m,n$ be positive integers such that $mp\neq0$ and $mp+n$ is even. Then
\begin{align}\label{eqn11}
\frac{1}{\sqrt{p}}&\sum_{j=0}^{p-1}\expo{\pi\ii(j^2m+jn)/p}=\notag\\
&=\expo{\pi\ii(mp-n^2)/4mp}\frac{1}{\sqrt{m}}\sum_{j=0}^{m-1}\expo{-\pi\ii(j^2p+jn)/m}.
\end{align}
\end{RFGQGS}
The definitions of the Gauss sums \eqref{eqn8} and \eqref{eqn10} as well as the Landsberg-Schaar's identity \eqref{eqn9} can be found in standard number theory books \cite{HardyWright:IntroTheoryNumbers,EverestWard:IntroNumberTheory,Nathanson:ElementaryMethodsNumberTheory}. The reciprocity formula for the generalized quadratic Gauss sum is not as well-known, and can be found in \cite{BerndtEvans81}.

\subsection{Explicit Solution\label{sct2C}}
We now show that a family of equientangled bases that interpolates continuously between the product basis and the maximally entangled basis exists for all dimensions $D$, as summarized by the following Theorem.
\begin{theorem}\label{thm1}
The collection of $D^2$ normalized states
\begin{align}\label{eqn12}
\ket{\psi_{m,n}(t)}&=\sum_{k=0}^{D-1}a_k(t)\ket{k\oplus m}{\ket{k\oplus m\oplus n}},\\
m,n &= 0,\ldots, D-1,\notag
\end{align}
 indexed by a real parameter 
\mbox{$t\in[0,1]$} with
\begin{equation}\label{eqn13}
a_k(t)=\frac{1}{D}\sum_{j=0}^{D-1}\ee^{\ii t\theta^0_j}\omega^{kj},
\qquad \omega=\ee^{2\pi\ii/D},
\end{equation}
with the particular choice of
\begin{equation}\label{eqn14}
\theta^0_j = \left\{ 
\begin{array}{l l}
  \pi  j^2/D & \quad \text{if $D$ is even}\\
  2\pi j^2/D & \quad \text{if $D$ is odd},
\end{array} \right.
\end{equation}
defines a family of equientangled bases that continuously interpolates between a product basis at $t=0$ and  a maximally entangled basis at $t=1$.
\end{theorem}
That \eqref{eqn12} defines a family of equientangled bases that consists of a product basis at $t=0$ follows directly from the remarks of Sec.~\ref{sct1}, $a_k(0)=\delta_{k,0}$. 
Next note that a continuous variation of $t$ in the interval $[0,1]$ corresponds to a continuous variation of the Schmidt coefficients of the states in the basis. The latter implies that no matter which measure one uses to quantify the entanglement, the measure will vary continuously with $t$, since any pure state entanglement measure depends only on the Schmidt coefficients of the state \cite{PhysRevLett.83.1046}.

The only thing left to show is that the basis states in Theorem~\ref{thm1} are maximally entangled when $t=1$, or, equivalently, that $|a_k(1)|=1/\sqrt{D}$ for all $k$. We prove this by explicitly evaluating the value of $a_k(1)$ in the following Lemma.
\begin{lemma}\label{lma1}
Let $a_k(t)$ and $\{\theta_j^0\}_{j=0}^{D-1}$ be as defined by Theorem~\ref{thm1}. Then for all $k$
\begin{equation}\label{eqn15}
a_k(1) = 
\frac{\expo{\pi\ii/4}}{\sqrt{D}}\times 
\left\{ 
\begin{array}{l l}
   \omega^{-k^2/2}, & \quad \text{if $D$ is even}\\
  \omega^{-k^2/4} \left(\frac{1-\ii^{2k+D}}{\sqrt{2}} \right), & \quad \text{if $D$ is odd}
\end{array} \right..
\end{equation}
\end{lemma}
Lemma~\ref{lma1} implies at once that $|a_k(1)|=1/\sqrt{D}$, and therefore proves Theorem~\ref{thm1}. 
\begin{proof}
(of Lemma~\ref{lma1}) Note first that the expression for $a_k(1)$ in \eqref{eqn13} with $\theta_j^0$ defined in \eqref{eqn14} resembles the generalized quadratic Gauss sum \eqref{eqn10}. We will use the  reciprocity formula \eqref{eqn11} to prove Lemma~\ref{lma1}. There are two cases to be considered: Even $D$ and odd $D$.

\textbf{Even $D$}. Note that one can rewrite $a_k(1)$ in \eqref{eqn13} with $\theta_j^0$ defined in \eqref{eqn14} as
\begin{equation}\label{eqn16}
a_k(1)=\frac{1}{D}\sum_{j=0}^{D-1}\expo{\pi\ii j^2/D} \; \expo{2\pi\ii j k/D}=\frac{1}{D}\sum_{j=0}^{D-1}\expo{\pi\ii(j^2+2 k j)/D}.
\end{equation}
Applying the reciprocity formula \eqref{eqn11} to last term in \eqref{eqn16} with $m=1,n=2k,p=D$ (noting that $mp+n=D+2k$ is even) yields 
\begin{align}\label{eqn17}
a_k(1)&=\frac{1}{\sqrt{D}}\ee^{\pi\ii (D-4k^2)/4D} = \frac{\ee^{\pi\ii/4}}{\sqrt{D}}(-1)^D \ee^{-\pi\ii k^2/D} \notag\\
&=\frac{\ee^{\pi\ii/4}}{\sqrt{D}}\omega^{-k^2/2},\text{ since $(-1)^D=1$ for even $D$}.
\end{align}

\textbf{Odd $D$}. The proof is essentially the same as in the even $D$ case, but we explicitly write it below for the sake of completeness. Using a similar argument we rewrite $a_k(1)$ in \eqref{eqn13} with $\theta_j^0$ defined in \eqref{eqn14} as
\begin{align}\label{eqn18}
a_k(1)=\frac{1}{D}\sum_{j=0}^{D-1}\expo{2\pi\ii j^2/D} \; \expo{2\pi\ii j k/D}=\frac{1}{D}\sum_{j=0}^{D-1}\expo{\pi\ii(2j^2+2 k j)/D}.
\end{align}
Applying again the reciprocity formula \eqref{eqn11} to last term in \eqref{eqn18} with $m=2,n=2k,p=D$ (noting that $mp+n=2D+2k$ is even) yields 
\begin{align}\label{eqn19}
a_k(1)&=\frac{1}{\sqrt{D}}\frac{\ee^{\pi\ii(2D-4k^2)/8D}}{\sqrt{2}}
(1+\ee^{-\pi\ii(D+2k)/2})\notag\\
&=\frac{\ee^{\pi\ii/4}}{\sqrt{D}}\omega^{-k^2/4} \left( \frac{1-\ii^{2k+D}}{\sqrt{2}} \right),
\end{align}
where we used $(-\ii)^{2k+D}=-(\ii^{2k+D})$ since $2k+D$ is odd.
This concludes the proof of Lemma~\ref{lma1} and implicitly of Theorem~\ref{thm1}.
\end{proof}

\subsection{Examples\label{sct2D}}
In this section we present some examples that illustrate the behaviour of the solution we provided in Theorem~\ref{thm1}, for various dimensions. First we consider $D=5$ and we plot the absolute values of the $a_k(t)$ coefficients as a function of $t$ in Fig.~\ref{fgr1}. It is easy to see that indeed the basis interpolates between a product basis and a maximally entangled one in a continuous manner. We observe that all coefficients are non-zero for $t>0$ and we believe that this is  probably also the case for all odd $D$'s.
\begin{figure}
\includegraphics[scale=0.45]{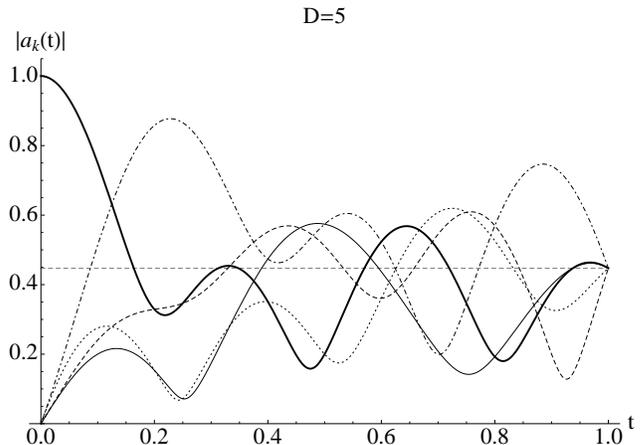}
\caption{The variation of $|a_k(t)|$ [dimensionless] with $t$ for $D=5$. Note how at $t=0$ all coefficients but one are zero, and how at $t=1$ all coefficients are equal in magnitude to $1/\sqrt{5}$, with a continuous variation in between. The horizontal dashed line represents the $1/\sqrt{5}$ constant function. }
\label{fgr1}
\end{figure}

In Fig.~\ref{fgr2} we perform the same analysis as above, but now for $D=8$. We observed that some coefficients vanish for some values of $t$, which seems to be true in general for even $D$.
\begin{figure}
\includegraphics[scale=0.45]{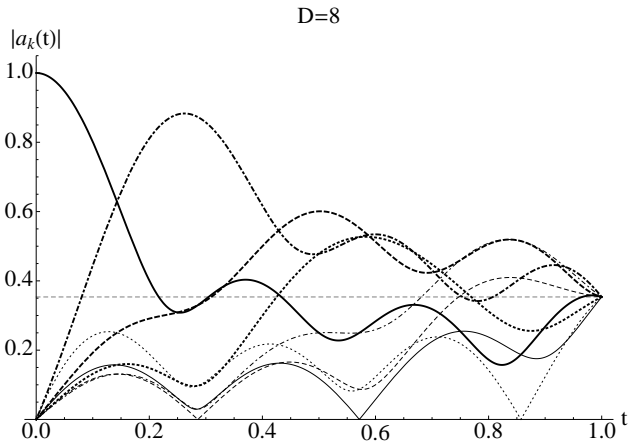}
\caption{The variation of $|a_k(t)|$ [dimensionless] with $t$ for $D=8$. Again note how at $t=0$ all coefficients but one are zero, and how at $t=1$ all coefficients are  equal in magnitude to $1/\sqrt{8}$, with a continuous variation in between. The horizontal dashed line represents the $1/\sqrt{8}$ constant function. }
\label{fgr2}
\end{figure}

In Fig.~\ref{fgr3} we plot the entropy of entanglement of the states in the basis as a function of $t$ for dimensions $D=2,3,5,8$ and $100$.
We see how the entanglement varies continuously but not monotonically between 0 and 1.
\begin{figure}
\includegraphics[scale=0.45]{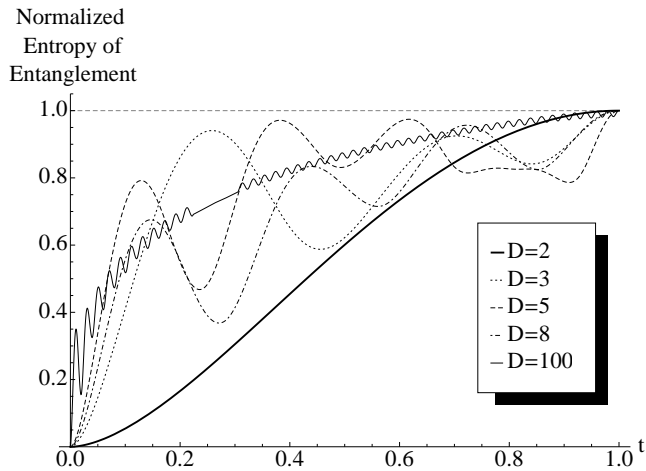}
\caption{The normalized entropy of entanglement [dimensionless, the logarithm is taken base $D$] as a function of $t$ for various dimensions. Note that the variation is not monotonic (except for $D=2$), although for large $D$ the oscillations tend to be smoothed out.}
\label{fgr3}
\end{figure}

Finally in Fig.~\ref{fgr4} we display a parametric plot  of the variation of the second Schmidt coefficient $a_1(t)$ in the complex plane as $t$ is varied from $0$ to $1$ for $D=51$, so that the reader can get an idea of how the coefficients defined in \eqref{eqn13} look in general. The other coefficients $a_k$ look similar. 

\begin{figure}
\includegraphics[scale=0.45]{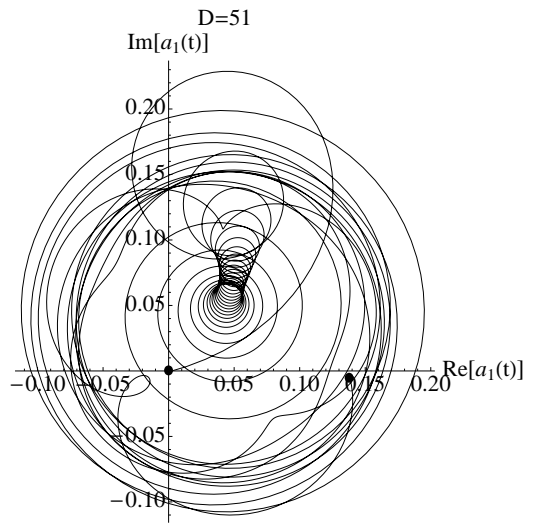}
\caption{Parametric plot of $a_1(t)$ [dimensionless] in the complex plane as $t$ is varied from $0$ to $1$. Note that $a_1(0)=0$ and $a_1(1)=\frac{1-\ii}{\sqrt{2\cdot 51}}\ee^{\pi\ii(\frac{1}{4}-\frac{1}{2\cdot 51})}$, the value provided by Lemma~\ref{lma1}. The starting point $t=0$ and the ending point $t=1$ are marked by solid disks.}
\label{fgr4}
\end{figure}

\section{Construction based on Graph States\label{sct3}}
\subsection{Explicit solution\label{sct3A}}
We provide below another solution to the problem that uses qudit graph states. 
Again having in mind a bipartite Hilbert space $\HC\otimes\HC$, both local spaces having dimension $D$, we define a one-qudit state
\begin{equation}\label{eqn20}
\ket{+}:=\frac{1}{\sqrt{D}}\sum_{k=0}^{D-1}\ket{k}.
\end{equation}
It is easy to see that the collection of $D$ states
\begin{equation}\label{eqn21}
\ket{\overline m}:= Z^m\ket{+},\quad {m=0,\ldots, D-1}
\end{equation}
defines an orthonormal basis of $\HC$ (also known as the Fourier basis), $\ip{\overline m}{\overline n} = \delta_{mn}$, where 
\begin{equation}\label{eqn22}
Z:=\sum_{k=0}^{D-1}\omega^k\dyad{k}{k},
\end{equation}
with  
$\omega=\expo{2\pi\ii/D}
$ being the $D$-th root of unity. It then follows at once that the collection of $D^2$ states
\begin{equation}\label{eqn23}
\ket{\overline m}\ket{\overline n}=(Z^m\otimes Z^n)\ket{+}\ket{+},
\quad{m,n=0,\ldots, D-1}
\end{equation}
defines an orthonormal product basis of the bipartite Hilbert space $\HC\otimes\HC$. 

Next we define the generalized controlled-Phase gate as
\begin{equation}\label{eqn24}
\CP:=\sum_{k=0}^{D-1} \dyad{k}{k} \otimes Z^k=\sum_{j,k=0}^{D-1}\omega^{jk}\dyad{j}{j}\otimes\dyad{k}{k}
\end{equation}
and note that $\CP$ is a unitary operator that commutes with $Z^m\otimes Z^n$, for all $m,n=0,\ldots,D-1$. The state
\begin{equation}\label{eqn25}
\ket{G}:=\CP\ket{+}\ket{+}=\frac{1}{D}\sum_{j,k=0}^{D-1}\omega^{jk}\ket{j}\ket{k}
\end{equation}
is an example of a two-qudit \emph{graph state} and it is not hard to see that $\ket{G}$ is maximally entangled. Then the collection of $D^2$ states
\begin{align}\label{eqn26}
(&Z^m\otimes Z^n)\ket{G}=(Z^m\otimes Z^n) \CP\ket{+}\ket{+}\notag\\
&=\CP (Z^m\otimes Z^n)\ket{+}\ket{+}, \quad{m,n=0,\ldots, D-1}
\end{align}
defines an orthonormal basis of the bipartite Hilbert space $\HC\otimes\HC$, which we call a \emph{graph basis}. Since $Z^m\otimes Z^n$ are local unitaries and $\ket{G}$ is a maximally entangled state, then all the other graph basis states must also be maximally entangled. For more details about graph states of arbitrary dimension see \cite{quantph.0602096, PhysRevA.78.042303, PhysRevA.81.032326}.

We now have all the tools to construct a continuous interpolating family of equientangled bases, as summarized by the Theorem below.
\begin{theorem}\label{thm2}
The collection of $D^2$ normalized states
\begin{align}\label{eqn27}
\ket{G_{m,n}(t)}&=(Z^m\otimes Z^n) \CP(t)\ket{+}\ket{+},\\
m,n&=0,\ldots, D-1\notag,
\end{align}
indexed by a real parameter $t\in[0,1]$ where 
\begin{equation}\label{eqn28}
\CP(t)=\sum_{j,k}\omega^{jkt}\dyad{j}{j}\otimes\dyad{k}{k}
\end{equation}
defines a family of equientangled bases that continuously interpolates between a product basis at $t=0$ and a maximally entangled basis at $t=1$. 
\end{theorem}
\begin{proof}
We make the crucial observation that $\CP(t)$ commutes with $Z^m\otimes Z^n$ for all $m,n=0,\ldots, D-1$ and all $t \in [0,1]$ which implies that $\{\ket{G_{m,n}(t)}\}_{m,n=0}^{D-1}$ defines an orthonormal basis since it differs from the orthonormal basis in \eqref{eqn23} only by the unitary operator $\CP(t)$. All states in the basis are equally entangled, and moreover, share the same set of Schmidt coefficients since any two basis states are equivalent up to local unitaries of the form $Z^m\otimes Z^n$.

Finally note that \mbox{$\CP(t=0)=I\otimes I$} and \mbox{$\CP(t=1)=\CP$} (defined in \eqref{eqn24}), and therefore at $t=0$ the basis is product, see \eqref{eqn23}, and at $t=1$ the basis is maximally entangled, see \eqref{eqn26}. The operator $\CP(t)$ can be viewed as a controlled-Phase gate whose ``entangling strength" can be tuned continuously. The Schmidt coefficients of the states in the basis vary continuously with $t$ and hence the entanglement also varies continuously with $t$, regardless of which entanglement measure one uses (see the remarks following Theorem~\ref{thm1}).
\end{proof}
Our construction above can be expressed in the framework described in the last two paragraphs of Sec. II of \cite{PhysRevA.73.012329}, by setting $U_m=Z^m$ and $V_n=Z^n$.

Next we prove that the Schmidt coefficients of the basis states in Theorem \ref{thm2} are all non-zero for any $t>0$, so all bases consist of full Schmidt rank states whenever $t>0$.

\begin{lemma}\label{lma2}
The equientangled family of bases $\{\ket{G_{m,n}(t)}\}_{m,n=0}^{D-1}$ defined in Theorem~\ref{thm2} consists of full Schmidt rank states, for any $0<t\leqslant1$.
\end{lemma} 
\begin{proof}
We will show that the product of the Schmidt coefficients is always non-zero, which implies that no Schmidt coefficient can be zero, whenever $0<t\leqslant 1$.

Let 
\begin{equation}\label{eqn29}
\ket{\psi}=\sum_{j,k=0}^{D-1}\Omega_{jk}\ket{j}\ket{k}
\end{equation}
be an arbitrary normalized pure state in a bipartite Hilbert space $\HC\otimes\HC$ and
let $\{\lambda_k\}$ denote the set of Schmidt coefficients of $\ket{\psi}$ satisfying $\sum_k{\lambda_k}=1$; note that they are equal to the squares of the singular values of the coefficient matrix  $\Omega$ in \eqref{eqn29}. The product of the squares of the singular values is just the product of the eigenvalues of $\Omega\Omega^\dagger$, the latter product being equal to $\det(\Omega\Omega^\dagger)=|\det(\Omega)|^2$, so we conclude that
\begin{equation}\label{eqn30}
\prod_{k=0}^{D-1}\lambda_k=|\det(\Omega)|^2.
\end{equation}

The states $\ket{G_{m,n}(t)}$ of Theorem~\ref{thm2} share the same set of Schmidt coefficients (they are all related by local unitaries) so it suffices to show that the product of the Schmidt coefficients is non-zero only for the state $\ket{G_{0,0}(t)}$. Recall that 
\begin{equation}\label{eqn31}
\ket{G_{0,0}(t)} = \CP(t)\ket{+}\ket{+} = \sum_{j,k}\frac{\omega^{jkt}}{D}\ket{j}\ket{k}.
\end{equation}
Expressing the coefficients $\omega^{jkt}/D$ as a matrix $\Omega(t)$, one can easily see that $D\cdot\Omega(t)$ is a $D\times D$ Vandermonde matrix whose determinant is 
\begin{equation}\label{eqn32}
\det\left[D\cdot\Omega(t)\right] = \prod_{j > k}(\omega^{jt}-\omega^{kt}) = \prod_{j>k}\omega^{kt}\left[\omega^{(j-k)t}-1\right]
\end{equation}
(see p.~29 of \cite{HornJohnson:MatrixAnalysis} for more details on Vandermonde matrices).
For a given $t$, the product above is zero if and only if at least one term is zero, i.e. there must exist integers $j$, $k$, with $0\leqslant k<j\leqslant D-1$, such that
\begin{equation}\label{eqn33}
(j-k)t=n D \Longleftrightarrow t=n \frac{D}{j-k},
\end{equation}
for some positive integer $n\geqslant 0$.
Note that $0<j-k\leqslant D-1$, so $D/(j-k)>1$ and the above equation can never be satisfied for $0<t\leqslant 1$. We have therefore proved that $\det[\Omega(t)]\neq0$ for $0<t\leqslant 1$, which, in the light of \eqref{eqn30}, is equivalent to saying that the product of the Schmidt coefficients is non-zero for $0<t\leqslant 1$, and this concludes the proof of the Lemma.
\end{proof}

For this family of equientangled bases, we do not have an analytic expression for the Schmidt coefficients nor the entropy of entanglement for general $D$ though they can be easily found by numerically diagonalizing the coefficient matrix $\Omega(t)\Omega(t)^\dagger$. Having said that, we derived a simple analytic expression for the \emph{product} of all Schmidt coefficients, see \eqref{eqn30} and \eqref{eqn32}, which is simply related to an entanglement monotone called $G$-concurrence (first introduced in \cite{PhysRevA.71.012318}) which is defined for a pure bipartite state \eqref{eqn29} in terms of its Schmidt coefficients $\{\lambda_k\}$ as
\begin{equation}\label{eqn34}
C_G(\ket{\psi}) = D\left(\prod_{k=0}^{D-1}\lambda_k\right)^{1/D} = D|\det(\Omega)|^{2/D}
\end{equation}
where $\sum_k\lambda_k=1$. The $G$-concurrence is zero whenever at least one Schmidt coefficient is zero and is equal to one if and only if the state is maximally entangled. Unlike the entropy of entanglement, we are able to show two analytical results that are true for all $D$ expressed in Lemmas \ref{lma3} and \ref{lma4}.

\begin{lemma}\label{lma3}
The $G$-concurrence of the equientangled basis states, $\{\ket{G_{m,n}(t)}\}_{m,n=0}^{D-1}$ defined in Theorem~\ref{thm2} is
\begin{equation}\label{eqn35}
C_G(t) = \frac{2^{D-1}}{D} \prod_{r=1}^{D-1}
\left[\sin^2(\pi r t/D)\right]^{(D-r)/D} \text{ for all $m,n,D$.}
\end{equation}
\end{lemma}

\begin{proof}
Every basis state has the same $G$-concurrence since they all share the same set of Schmidt coefficients (recall that they differ only by local unitaries). Invoking the definition of $G$-concurrence, we have
\begin{align}
C&_G(t)=D|\det(\Omega(t))|^{2/D}
=D\left|\frac{1}{D^D}\det(D\cdot\Omega(t))\right|^{2/D}\label{eqn36}\\
&=\frac{1}{D}\prod_{j>k}\left|\omega^{(j-k)t}-1\right|^{2/D}=\frac{1}{D}\prod_{r=1}^{D-1}\left|\omega^{rt}-1\right|^{\frac{2(D-r)}{D}}\label{eqn37}\\
&=\frac{1}{D}\prod_{r=1}^{D-1}\left[2 - 2\frac{\omega^{rt}+\omega^{-rt}}{2} \right]^{\frac{D-r}{D}}\label{eqn38}\\
&=\frac{1}{D}\prod_{r=1}^{D-1}2^{(D-r)/D}\left[1-\cos(2\pi r t/D)\right]^{\frac{D-r}{D}}\label{eqn39}\\
&=\frac{1}{D}\prod_{r=1}^{D-1}2^{2(D-r)/D}\left[\sin^2(\pi r t/D)\right]^{\frac{D-r}{D}}\label{eqn40}\\
&=\frac{2^{D-1}}{D}\prod_{r=1}^{D-1}\left[\sin^2(\pi r t/D)\right]^{(D-r)/D} \label{eqn41},
\end{align}
where in \eqref{eqn36} we used the fact that $\det(cM)=c^D\det(M)$ for a $D\times D$ arbitrary matrix $M$ and an arbitrary constant $c$. The first equality in \eqref{eqn37} follows at once from \eqref{eqn32}, whereas the second equality in \eqref{eqn37} follows from a simple counting argument in which one replaces $j-k$ by $r$, making sure that the different pairs $(j,k)$, $j>k$, that give rise to the same 
$r$ are counted; for a given $r$ there are $D-r$ such pairs.
\end{proof}

It turns out the $G$-concurrence has the following nice property:
\begin{lemma}\label{lma4}
The $G$-concurrence of the basis states $\{\ket{G_{m,n}(t)}\}_{m,n=0}^{D-1}$ in \eqref{eqn35} is strictly increasing in the open interval $t \in (0,1)$, for all dimensions $D$.
\end{lemma}
\begin{proof}
We prove this by showing that the first derivative of the $C_G(t)$ with respect to $t$ is strictly positive. First note that since $C_G(t)>0$ is a positive function in the interval $t \in (0,1)$. This means showing $\frac{\mathrm{d}}{\mathrm{d} t} C_G(t) > 0$ is equivalent to showing that $\frac{\mathrm{d}}{\mathrm{d} t} \log C_G(t) > 0$. The derivative of the logarithm of \eqref{eqn41} is
\begin{equation}\label{eqn42}
\frac{\mathrm{d}}{\mathrm{d} t} \log C_G(t) = \frac{2\pi}{D^2} \sum_{r=1}^{D-1}r(D-r)
\cot(\pi r t/D),
\end{equation}
where $\cot(\cdot)$ denotes the cotangent function. We differentiate again to get
\begin{equation}\label{eqn43}
\frac{\mathrm{d^2}}{\mathrm{d} t^2}\log C_G(t)=-\frac{2\pi^2}{D^3}\sum_{r=1}^{D-1}r^2(D-r)\frac{1}{\sin^2(\pi r t/D)}.
\end{equation}
Note that the right hand side of \eqref{eqn43} is strictly negative whenever $t>0$, which implies that the first derivative of the logarithm \eqref{eqn42} is a strictly decreasing function of $t$ and hence achieving its minimum value at $t=1$, which is given by
\begin{equation}\label{eqn44}
\frac{\mathrm{d}}{\mathrm{d} t} \log C_G(t)\mid_{t=1} = \frac{2\pi}{D^2}\sum_{r=1}^{D-1}r(D-r) \cot(\pi r/D)=0,
\end{equation}
where the last equality follows from symmetry considerations (terms cancel one by one). We have shown $\frac{\mathrm{d}}{\mathrm{d} t} \log C_G(t) > 0 \Leftrightarrow \frac{\mathrm{d}}{\mathrm{d} t} C_G(t) > 0$ and therefore we conclude that the $G$-concurrence is strictly increasing for $t \in (0,1)$.
\end{proof}

\subsection{Extension to multipartite systems\label{sct3C}}
The construction presented in Theorem~\ref{thm2} can be easily generalized to  multipartite systems of arbitrary dimension. The concept of maximally entangled states is not defined for three parties or more, so in this case, the family continuously interpolates between a product basis and a qudit graph basis. It is still true that for a fixed $t$, all basis states constructed this way have the same entanglement (as quantified by any entanglement measure) since they only differ by local unitaries. 

As a specific example, consider the tripartite GHZ state \mbox{$(\ket{000}+\ket{111})/\sqrt{2}$}. This state is a stabilizer state \cite{NielsenChuang:QuantumComputation} and therefore is local-unitary equivalent \cite{quantph.0111080}  to a graph state $\ket{G}=(\CP_{12} \CP_{23} \CP_{13}) \ket{+}_1 \ket{+}_2 \ket{+}_3$, where the subscripts on $\CP$ indicates which pair of qubits the $\CP$ gate is applied to. By varying the ``strength" of the controlled-Phase gate, one can now construct a family of equally entangled basis for the Hilbert space of 3 qubits that continuously interpolates between a product basis and the GHZ-like graph basis. This GHZ construction can be easily generalized to higher dimensions and also to $n$ parties by using the complete graph given by
\begin{equation}
\label{eqn45}
\left| G_{\mathrm{GHZ}}(t) \right \rangle := \prod_{i=1}^{n-1} \prod_{j>i}^n \CP_{ij}(t) \; \ket{+}^{\otimes n}
\end{equation}
where $\ket{+}$ is defined in \eqref{eqn20} and $\CP(t)$ is defined in \eqref{eqn28}. Finally note that this construction works for any graph state of any dimension, and not just for GHZ-like graph states. Such bases with tunable entanglement may be of use in the study of multipartite entanglement.

\subsection{Examples\label{sct3B}}
In this subsection we perform a similar analysis as the one in Sec.~\ref{sct2D}, so that one can easily compare the behaviour of both solutions.

We consider again a $D=5$ example, for which we plot in Fig.~\ref{fgr5} the square root of the Schmidt coefficients as functions of $t$. It is easy to see that indeed the basis interpolates between a product basis and a maximally entangled one in a continuous manner. As proven in Lemma~\ref{lma2}, all Schmidt coefficients are non-zero for $t>0$.
\begin{figure}
\includegraphics[scale=0.5]{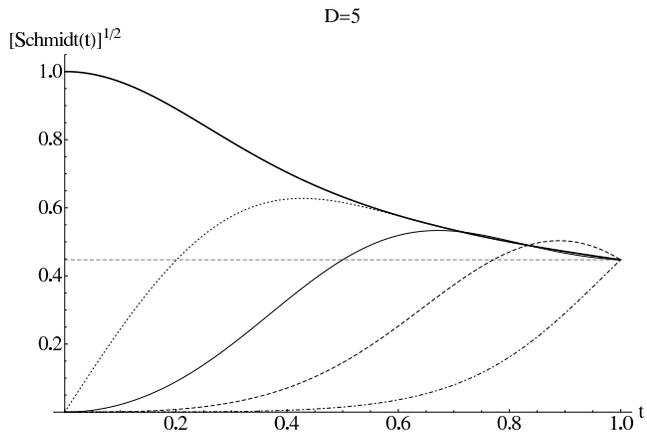}
\caption{The square roots of the Schmidt coefficients [dimensionless] as functions of $t$ for $D=5$. Note how at $t=0$ all coefficients but one are zero, and how at $t=1$ all coefficients are equal in magnitude to $1/\sqrt{5}$, with a continuous variation in between. The horizontal dashed line represents the $1/\sqrt{5}$ constant function. }
\label{fgr5}
\end{figure}

In Fig.~\ref{fgr6} we plot the same quantities for $D=8$. 
\begin{figure}
\includegraphics[scale=0.5]{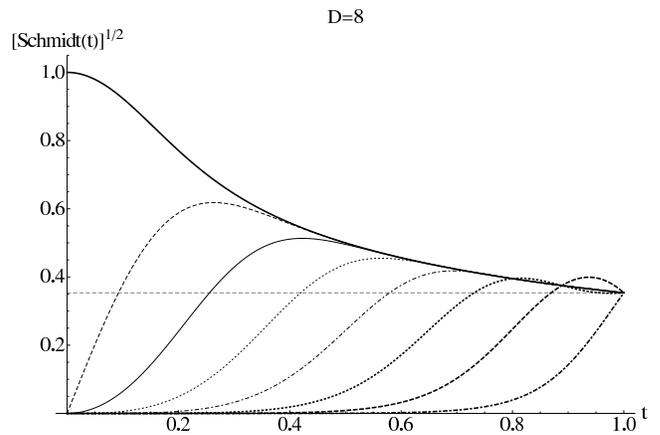}
\caption{The square roots of the Schmidt coefficients [dimensionless] as functions of $t$ for $D=8$. Again note how at $t=0$ all coefficients but one are zero, and how at $t=1$ all coefficients are  equal in magnitude to $1/\sqrt{8}$, with a continuous variation in between. The horizontal dashed line represents the $1/\sqrt{8}$ constant function. }
\label{fgr6}
\end{figure}
Observe how in both examples above the variation of the Schmidt coefficients is not oscillatory, as in the examples of Sec.~\ref{sct2D}.

\begin{figure}
\includegraphics[scale=0.45]{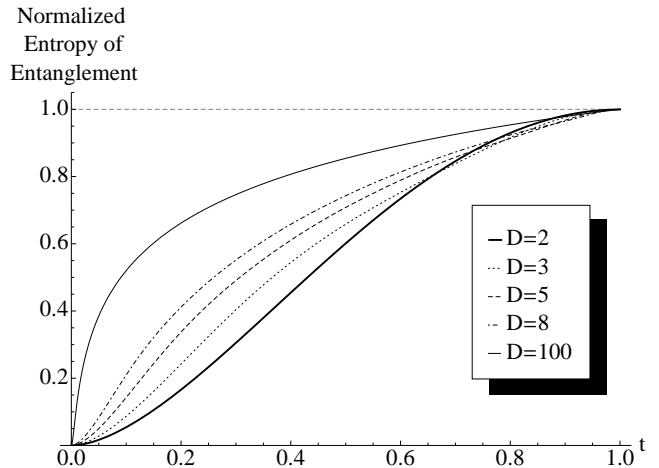}
\caption{The normalized entropy of entanglement [dimensionless, the logarithm is taken base $D$] as a function of $t$ for various dimensions. Note that the variation seems to be monotonically increasing for all $D$, a statement we did not prove.}
\label{fgr7}
\end{figure}

\begin{figure}
\includegraphics[scale=0.45]{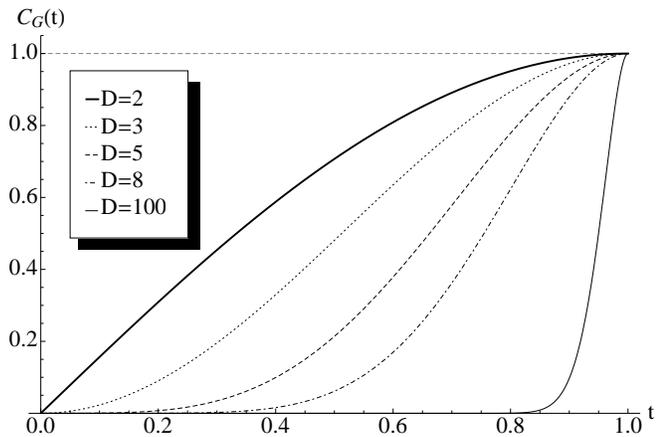}
\caption{The $G$-concurrence [dimensionless] as function of $t$ for various dimensions. The variation is strictly increasing in $t$ for all $D$ as shown in Lemma \ref{lma4}.}
\label{fgr8}
\end{figure}

In Fig.~\ref{fgr7} we plot the entropy of entanglement of the basis states as a function of $t$ for dimensions $D=2,3,5,8$ and $100$. We observe that the entropy of entanglement varies continuously and monotonically between 0 and 1. It is not known if the entropy of entanglement is always strictly increasing for all $D$ although we verified
this by visual inspection for all $D\leqslant 10$. In Fig.~\ref{fgr8}, we plot the $G$-concurrence for the same dimensions. We see how the curves are strictly increasing and this is true for all $D$ as proven in Lemma~\ref{lma4}.

\section{Conclusion\label{sct4}}
We have solved the problem posed in \cite{PhysRevA.73.012329} by providing two families of equientangled bases for two identical qudits for arbitrary dimension $D$. The construction of the first solution is based on quadratic Gauss sums and follows along the lines of \cite{PhysRevA.73.012329}, whereas the second family is constructed using a different method based on qudit graph states.

The first solution based on quadratic Gauss sums has an explicit analytic expression for the Schmidt coefficients that is easy to evaluate since they are just sums with $D$ terms (see \eqref{eqn13} and \eqref{eqn14}). However some Schmidt coefficients can be zero and the entropy of entanglement of the states in the basis varies non-monotonically with $t$ for $D>2$. 

The second solution based on graph states consists entirely of full Schmidt rank states for any $0<t\leqslant 1$ that seem to have an entropy of entanglement that is strictly increasing as $t$ increases. Unfortunately we did not find a simple analytic expression for the Schmidt coefficients, but they can be computed numerically without much difficulty. We found a simple analytic expression for another pure state entanglement measure, the $G$-concurrence, which we proved is strictly increasing as $t$ increases. Finally we remark that one can extend this construction to equally entangled bases of more than two parties that interpolate continuously between a product basis and a graph basis even if the concept of maximally entangled states is not defined for more than two parties. This construction may be of interest in studying multipartite entanglement.

\begin{acknowledgments}
The research described here received support from the National Science Foundation through Grant No. PHY-0757251.
\end{acknowledgments}


%

\end{document}